\documentclass[11pt]{llncs}
\usepackage{xspace}
\usepackage{amsmath}
\usepackage{amssymb}
\usepackage{amsfonts}
\usepackage{fullpage,times}

\newtheorem{observation}[theorem]{Observation}

\newcommand{\algo}[1]{\ensuremath{\mbox{\textsc{#1}}}\xspace}

\newcommand{\alg}{\algo{Alg}}

\newcommand{\opt}{\algo{opt}}
\newcommand{\optn}{\algo{opt-nice}}
\newcommand{\sol}{\algo{Sol}}

\newcommand{\aux}{\algo{aux}}

\newcommand{\pr}{\mbox{{\sc{gebp}}}}

\newcommand{\pra}{\mbox{{\sc{ebp}}}}
\newcommand{\prb}{\mbox{{\sc{ebp-ubs}}}}
\newcommand{\prc}{\mbox{{\sc{gebp-bpv}}}}

\newcommand{\R}{{\cal R}}

\newcommand{\eps}{\varepsilon}

\begin{document}

\title{Approximation schemes for the generalized extensible bin packing problem\thanks{This research was supported by a grant from the GIF, the German-Israeli Foundation for Scientific Research and Development (grant number I-1366-407.6/2016) and by a grant from the ISF, the Israel Science Foundation (grant number 308/18).}}

\author{Asaf Levin\inst{1}}
\institute{Faculty of
Industrial Engineering and Management, The Technion, 32000 Haifa,
Israel. \email{levinas@ie.technion.ac.il}. }\date{}

\maketitle


\begin{abstract} 
We present a new generalization of the extensible bin packing with unequal bin sizes problem.  In our generalization the cost of exceeding the bin size depends on the index of the bin and not only on the amount in which the size of the bin is exceeded.  This generalization does not satisfy the assumptions on the cost function that were used to present the existing polynomial time approximation scheme (PTAS) for the extensible bin packing with unequal bin sizes problem.  In this work, we show the existence of an efficient PTAS (EPTAS) for this new generalization and thus in particular we improve the earlier PTAS for the extensible bin packing with unequal bin sizes problem into an EPTAS.  Our new scheme is based on using the shifting technique followed by a solution of polynomial number of $n$-fold programming instances.  In addition, we present an asymptotic fully polynomial time approximation scheme (AFPTAS) for the related bin packing type variant of the problem.
\end{abstract}

\section{introduction}
We define the following load balancing on parallel machines problem that we name the {\em generalized extensible bin packing problem} (\pr).  The input consists of $n$ jobs, where job $j$ has size $p_j\geq 0$, there are $m$ machines where for all $i$, machine $i$ is associated with three positive input numbers $f_i,c_i,\sigma_i$, such that the following assumption holds: 
\begin{equation}
f_i=c_i\cdot \sigma_i, \ \ \ \forall i .  \label{assump1}
\end{equation} 
Assigning the set of jobs $S_i$ to machine $i$, incurs a load on machine $i$ that is the total size of jobs in $S_i$.  That is, the load of machine $i$ that is assigned the set of jobs $S_i$ is $L_i=\sum_{j\in S_i} p_j$, and the cost of machine $i$ is 
\begin{equation*}
cost_i(L_i)=\left\{ \begin{array}{l l l}
f_i & \mbox{ \ \ \ } & \mbox{if } L_i\leq c_i\\
f_i+\sigma_i\cdot (L_i-c_i) & & \mbox {otherwise}
\end{array} \right .  \end{equation*}
The goal of \pr\ is to find a partition of the jobs to $m$ machines such that the total cost of the machines (in this solution) is minimized.  In this definition of the cost function of machine $i$, the value of $f_i$ is seen as a {\em fixed cost of machine $i$}, the value of $c_i$ is the {\em standard capacity of machine $i$}, and $\sigma_i$ is the {\em cost of extending the capacity of machine $i$ by one unit of overtime}.  The value of $\frac{1}{\sigma_i}$ captures the speed in which increasing the total size of jobs assigned to $i$ causes the cost of $i$ to increase by one unit.  This speed is similar to the roles of speeds in the environment of uniformly related machines that is widely studied in the scheduling literature.  In our study it will be easier to refer to the reciprocal of the speed (i.e., to the values of $\sigma_i$) and not to the speeds.

The extensible bin packing problem (\pra) is the special case of \pr\ where for every machine $i$ we have $f_i=c_i=\sigma_i =1$ (note that these values satisfy (\ref{assump1})).  Even this special case is strongly NP-hard via the standard reduction from 3-Partition. This extensible bin packing problem was suggested by \cite{speranza+,dell+}. 
Another special case of \pr\ that was considered before is the case of extensible bin packing with unequal bin sizes (\prb).  This \prb\ is defined as the special case of \pr\ where for every machine $i$, $\sigma_i=1$ and $f_i=c_i$ (once again  (\ref{assump1}) holds for such values).  Another interesting special case of \pr\ that generalizes \pra\ is the generalization from identical machines to uniformly related machines, that is, the special case of \pr\ where $f_i=1$ for all $i$.  Observe that this last case does not generalizes \prb\ and it is not a special case of \prb.  Our new model is defined in order to generalizes all these special cases.

Before we state our main result and present the literature, we define the notion of
approximation algorithms and the different types of approximation
schemes.  An $\R$-approximation algorithm for a minimization
problem is a polynomial time algorithm that always finds a
feasible solution of cost at most $\R$ times the cost of an
optimal solution. The infimum value of $\R$ for which an algorithm
is an $\R$-approximation is called the approximation ratio of the algorithm. A polynomial time
approximation scheme (PTAS) is a family of approximation
algorithms such that the family has a $(1+\eps)$-approximation
algorithm for any $\eps >0$. An efficient polynomial time
approximation scheme (EPTAS) is a PTAS whose time complexity is of
the form $f(\frac{1}{\eps}) \cdot poly(n)$ where $f$ is some (not
necessarily polynomial) computable function, and $poly(n)$ is a polynomial
function of the length of the (binary) encoding of the input. A
fully polynomial time approximation scheme (FPTAS) is a stronger
concept, defined like an EPTAS, but the function $f$ must be a
polynomial in $\frac 1{\eps}$. When we consider an
EPTAS we say that an algorithm (for some problem) has a
{\it polynomial running time complexity} if its time complexity is of
the form $f(\frac{1}{\eps}) \cdot poly(n)$. Note that while a PTAS
may have time complexity of the form $n^{g(\frac{1}{\eps})}$,
where $g$ can be polynomial or even super-exponential, this cannot be the
case for an EPTAS.  The notion of an EPTAS is modern and finds its
roots in the FPT (fixed parameter tractable) literature (see
\cite{CT97,DF99,FG06,Marx08}). It was introduced in order to
distinguish practical from impractical running times of PTAS's,
for cases where a fully polynomial time approximation scheme
(FPTAS) does not exist (unless P=NP). In this work, we design an
EPTAS for \pr\  for which an FPTAS does not exist unless P=NP as \pr\ is strongly NP-hard.

 In \cite{speranza+} Speranza and Tuza analyzed an online variant of \pra\ and considered the list scheduling heuristic showing that it is a $5/4$-approximation while a slightly improved algorithm is suggested whose approximation ratio is $1.228$ and a lower bound of $7/6$ is established for this online variant.  In \cite{dell+} Dell'Olmo et al.~showed that the longest processing time heuristic is a $13/12$-approximation for \pra.  The EPTAS of Alon et al. \cite{AAWY97,AAWY98} for load balancing on identical machines solves \pra\ and thus this special case admits an EPTAS prior to this work.  The time complexity of this EPTAS for \pra\ (among other problems on identical machines) was improved in the work of Jansen, Klein, and Verschae \cite{JKV16}. The online problem was studied further in  \cite{ye2003line1}.  
See also \cite{denton2010optimal,berg2017fast,sagnol2018price} for a study of this special case in the stochastic settings in the context of scheduling operating rooms, and \cite{sirdey2007combinatorial} for a use of the approximation algorithms for this problem in PCM interface management arising in wireless switch design.

The study of \prb\ was initiated by Dell'Olmo and Speranza \cite{dell1999} who showed that the approximation ratio of the longest processing time heuristic is $4-2\sqrt{2}$ and that the approximation ratio of the online algorithm list scheduling is exactly $\frac 54$.  They also showed that any online algorithm has an approximation ratio of at least $\frac 76$. 
The PTAS of Epstein and Tassa \cite{ET06} for vector scheduling in asymmetric settings gives a PTAS for \prb.  Their assumption that the cost functions of the machines have a common constant upper bound on the Lipschitz constants cannot be met for \pr\ as this means that the maximum ratio between the costs of extending a pair of machines by one unit of overtime is bounded by a constant (i.e., their scheme assumes that $\frac{\sigma_i}{\sigma_{i'}}$ is bounded by a constant  independent of the pair of machines $(i,i')$).  The online problem of \prb\ was also studied in \cite{ye2003line} who analyzed the performance guarantee of list scheduling as a function of the standard capacities of the machines and present an improved online algorithm for the cases $m=2,3$.

Thus, with respect to the existence of approximation schemes, \pra\ was known to admit an EPTAS while \prb\ was known to have a PTAS (that is not an EPTAS).  The approximability of \pr\ as well as its special case of $f_i=1$ for all $i$ were not studied before.

Our main result is an EPTAS for \pr.  In particular, we improve upon the scheme of \cite{ET06} for \prb\ and present the first EPTAS for this (previously studied) special case.  Our scheme first apply preprocessing steps and then breaks the asymmetry between the machines in a two steps approach.  In the first step, we use the machinery of the {\em shifting technique} in order to partition the instance into polynomially many sub-instances each of which has the additional property that the standard capacities of the machines in the sub-instance are similar.  The resulting sub-instance still captures unbounded asymmetry between the machines, and in order to tackle the sub-instances we use the recent algorithms for $n$-fold programming.  We refer to \cite{JKMR19} for an earlier EPTAS for a different scheduling problem that is based on solving an $n$-fold programming instance. The time complexity of our scheme is a single exponential function of $\frac 1{\eps}$ times a polynomial of $n,m$.

We conclude this work in Section \ref{prcsec} by showing the existence of an asymptotic fully polynomial time approximation scheme for a related bin packing type variant of the problem similarly to the variant of \pra\ studied by \cite{coffman2001,coffman2006}.   In this variant of \pr\ the number of machines of each type is not part of the input and is determined as part of the solution.  Namely, for every $i=1,2,\ldots ,m$ we first decide how many machines of type $i$ with a fixed cost $f_i$, a standard capacity $c_i$, and the cost $\sigma_i$ for overtime, to have where for all $i$, $f_i=c_i \cdot \sigma_i$.  In a second stage we find a feasible allocation of the jobs to the machines we have (according to the decisions made in the first stage). We denote this variant \prc. In Section \ref{prcsec} we establish the existence of an asymptotic fully polynomial time approximation scheme (AFPTAS) for \prc.  We note that the special case of one type of machines with $f_1=c_1=\sigma_1=1$ was considered by Coffman and Lueker \cite{coffman2001,coffman2006} who presented an AFPTAS for this special case of \prc.

\paragraph{Paper outline.}  We present our EPTAS for \pr\ in the main part of the paper.  This exposition is partitioned into preprocessing steps and characterization of near optimal solutions in Section \ref{initialsec}, followed by an analysis of the shifting technique when applied to \pr\ in Section \ref{shiftsec}, and finally the use of the $n$-fold programming algorithm to solve the family of sub-instances resulting from the shifting step is described in Section  \ref{nfoldsec}.  We establish the existence of an AFPTAS for \prc\ in Section  \ref{prcsec}.

\section{Preprocessing steps and the structure of a near optimal solution\label{initialsec}}
We assume that $\eps>0$ satisfies that $\frac 1{\eps}$ is integer.  We use the fact that in order to establish the existence of an EPTAS for \pr, it suffices to establish for some integer constant $z$, a $(1+\eps)^{z}$-approximation algorithm whose time complexity is upper bounded by the product of a computable function of $\frac{1}{\eps}$ and a polynomial of the input length.  When we state time complexity of steps in our algorithm we ignore polynomial factors of $\frac 1{\eps}$.

Our preprocessing steps consists of scaling and rounding of the input parameters.
Our goal is to assume that job sizes are rounded and that the minimum value of $\sigma_i$ is $1$.

First, we consider the scaling of the parameters that allows us to assume that $\min_i \sigma_i =1$.  That is, we prove the following lemma.
\begin{lemma}\label{sigma-min}
Without loss of generality, $\min_i \sigma_i=1$.
\end{lemma}
\begin{proof}
Assume that the input of \pr\ does not satisfy the claim.  Then we let $\sigma=\min_i \sigma_i$, and we do the following.  For every machine $i$ we multiply $c_i$ by $\sigma$, and we divide $\sigma_i$ by $\sigma$.  In addition, for every job $j$, we multiply the job size $p_j$ by $\sigma$.  Observe that the new input satisfies $f_i=c_i\cdot \sigma_i$ for all $i$.

Next, we show that for every solution for the original input, the cost of the solution in the new input is the same as it was in the original input.  To see this fact, note that for every machine $i$, its load, i.e., the value of $L_i$ is $\sigma$ times its value in the solution for the original input, and thus it satisfies $L_i\leq c_i$ in the new input if and only if it is satisfied in the original input.  Furthermore, the value of $L_i-c_i$ in the new input is exactly $\sigma$ times its value in the original input, and thus the cost of machine $i$ is the same in the two inputs.
\qed \end{proof}

Next, without loss of generality, we assume that  machines are sorted according to their  standard capacities, that is, we assume $c_1\geq c_2 \geq \cdots \geq c_m$.

Throughout this work we use the following observation.
\begin{observation}\label{obs1}
Let $x,y$ be two numbers such that $x\leq y\leq (1+\eps)x$ and let $i$ be a machine, then $cost_i(x)\leq cost_i(y) \leq (1+\eps)\cdot cost_i(x)$.
\end{observation}
\begin{proof}
The inequality $cost_i(x)\leq cost_i(y)$ holds by the fact that the cost function $cost_i$ is monotone non-decreasing (as $\sigma_i$ is at least $1$, and thus non-negative).  The inequality $cost_i(y) \leq (1+\eps)\cdot cost_i(x)$ holds by the following argument.  If $y<c_i$, then $x,y<c_i$ and we have $cost_i(x)=cost_i(y)$ and the inequality holds.  Otherwise, $y\geq c_i$ and by the assumption $f_i=c_i\cdot \sigma_i$, we conclude $cost_i(x) \geq \sigma_i \cdot x$ and so $cost_i(x)\geq \sigma_i \cdot x \geq \sigma_i \cdot \frac{y}{1+\eps}  = \frac{1}{1+\eps} \cdot cost_i(y)$ establishing the required inequality.
\qed
\end{proof}

Next, we consider the rounding of the jobs sizes and we use the following rounding method.  This rounding method is motivated by the fact that the $n$-fold programming formulations which we use to solve sub-instances of the rounded problem later on, assume that all coefficients of the constraint matrix are (relatively small) integers.  Thus, for every job $j$, we let $\tau(j)$ be the integer value such that $p_j \in \left[ \frac{1}{\eps^{\tau(j)}}, \frac{1}{\eps^{\tau(j)+1}} \right)$.  We let the rounded value of $p_j$ be 
\begin{equation*}
p'_j=\left\lceil \frac{p_j}{(1/\eps)^{\tau(j)-1}} \right\rceil \cdot \left( (1/\eps)^{\tau(j)-1} \right) . 
\end{equation*}
The rounded instance $I'$ is the instance of \pr\ in which the values of the input parameters are $c_i,\sigma_i,f_i$ for all $i$ (such that $\min_i \sigma_i=1$), and $p'_j$ for all $j$.
In the sequel, we use the fact that in $I'$, for every integer value of $\tau$ there are at most $1/\eps^2$ distinct rounded job sizes in the interval $\left[ \frac{1}{\eps^{\tau}}, \frac{1}{\eps^{\tau+1}} \right)$.  

For a solution \sol\ and an instance $\hat{I}$ of \pr, we denote by $cost(\sol,\hat{I})$ its objective function value where the input parameters are according to $\hat{I}$  (in particular we use this notation for $\hat{I}=I$ and for $\hat{I}=I'$).  Next we analyze the impact of the rounding step on the performance guarantee of our algorithm. 
\begin{lemma}
Let \sol\ be a feasible solution.  Then, $cost(\sol,I)\leq cost(\sol,I')\leq (1+\eps)\cdot cost(\sol,I)$.
\end{lemma}
\begin{proof}
For every job $j$, we have $p'_j \geq p_j$ as we round up the size of $j$. For every $j$, we have $p'_j \leq (1+\eps)p_j$ as $\eps p_j \geq (1/\eps)^{\tau(j)-1}$ by definition of $\tau(j)$, and when we round up $p_j$ we increase its size by at most $(1/\eps)^{\tau(j)-1}$, that is, we have $$ p'_j \leq \left( \frac{p_j}{(1/\eps)^{\tau(j)-1}} +1 \right) \cdot \left( (1/\eps)^{\tau(j)-1} \right)  \leq p_j + \left( (1/\eps)^{\tau(j)-1} \right) \leq (1+\eps)p_j$$ as we argued.

For every machine $i$, the total size of jobs assigned to $i$ in $\sol$ as a solution  to $I$ is at most the total size of jobs assigned to $i$ in $\sol$ as a solution to $I'$ and this is at most $(1+\eps)$ times the total size of jobs assigned to $i$ as a solution to $I$.  The claim follows by Observation \ref{obs1} and summing the costs of all machines.
\qed
\end{proof}
In what follows, we assume that the original input of \pr\ satisfies the assumption of Lemma \ref{sigma-min} and the job sizes are already rounded as they are in $I'$.  With a slight abuse of notation, we let $f_i,c_i,\sigma_i$ be the parameters of machine $i$ and $p_j$ be the size of job $j$ (for all $i,j$) in this input which we denote by $I$.  It is sufficient to provide an EPTAS for $I$, and this would imply an EPTAS for the original input.

Next, we characterize near optimal solutions.  We let $h$ be an index of a machine for which $\sigma_h=1$.  
\begin{definition}
A solution \sol\ for \pr\ is called a {\em nice solution} if for every machine $i\neq h$, the total size of jobs assigned to $i$ (in \sol) is at most $\frac{c_i}{\eps}$.
\end{definition}
The proof of the following lemma uses the observation that since $\sigma_h=1$, adding a set of jobs of total size $x$ to machine $h$ increases the cost of $h$ by at most $x$.

\begin{lemma}
Let \opt\ be an optimal solution (for the rounded instance) whose cost is denoted as $cost(\opt)$.  Then, there is a nice solution $\sol_{nice}$ whose cost $cost(\sol_{nice})$ is at most $(1+\eps)\cdot cost(\opt)$. 
\end{lemma}
\begin{proof}
Consider the solution \opt.  Let $S$ be the set of machines $i$ such that $i\neq h$ and \opt\ assigns to $i$ a set of jobs $O_i$ of total size larger than $\frac{c_i}{\eps}$.  We create the solution $\sol_{nice}$ by changing the assignment of jobs in $\cup_{i\in S} O_i$ so that these jobs are assigned to $h$, while the assignment of other jobs is the same as in \opt.  That is, we reassign the jobs that were assigned to machine in $S$ so that they are assigned to $h$.  Observe that the new solution $\sol_{nice}$ is indeed a nice solution.
 
Next, we upper bound the cost of the new solution.  Let $L_i$ be the total size of jobs assigned by \opt\ to machine $i$ (for all $i$).  We have the following.
\begin{eqnarray*}
cost(\sol_{nice}) &\leq & \sum_{i\notin S} cost_i(L_i) +\sigma_h \cdot \sum_{i\in S} L_i + \sum_{i\in S} f_i \\
&\leq & \sum_{i\notin S} cost_i(L_i) + \sum_{i\in S} (L_i\sigma_i + f_i) \\
&=   &  \sum_{i} cost_i(L_i) + \sum_{i\in S} f_i
\end{eqnarray*}
however, $\sum_{i} cost_i(L_i) = cost(\opt)$, and thus it suffices to show that for every $i\in S$, we have 
$f_i \leq  \eps \cdot cost_i(L_i) $.  This follows as $f_i=c_i\cdot \sigma_i \leq \eps \cdot L_i \cdot \sigma_i = \eps \cdot cost_i(L_i)$ where the inequality holds as $i\in S$ and the last equation holds by the definition of $cost_i$ as $L_i\geq c_i$.
\qed \end{proof}

We let $\optn$ be an optimal solution among all nice solutions for the rounded instance whose cost is $cost(\optn)$.  In what follows we show an algorithm that returns a solution \sol\ whose cost is at most $(1+\eps) ^4(1+3\eps)\cdot cost(\optn)$ and its time complexity is $T'(n,m,\frac{1}{\eps})$.  This solution \sol\ is a feasible solution to the original instance of \pr\ that is obtained in time $T'(n,m,\frac{1}{\eps})+n+m$ (as the rounding takes $O(n)$ and the scaling takes $O(n+m)$) whose approximation ratio as an approximation algorithm for the original instance of \pr\ is $(1+\eps)^6 \cdot (1+3\eps)$.

\section{Using the shifting technique to obtain a family of instances with machines with similar standard capacities\label{shiftsec}}
We use the shifting technique \cite{HM85,HocBook} to partition the rounded instance into a family of problems that can be solved almost independently.
For every value of $t=0,1,\ldots, \frac{1}{\eps}-1$ we let 
\begin{equation*}
M(t) = \{ i\in \{1,2,\ldots ,m\} : \lceil \log_{1/\eps^2} c_i \rceil \equiv t  \mod \frac{1}{\eps} \} ,
\end{equation*}
and we let $t_{min}$ be an index such that \begin{equation*} t_{min} \in \arg\min_t \sum_{i\in M(t)} f_i . \end{equation*}
Observe that the value $t_{min}$ is easily computed in $O(m)$ time.

Recall that $\optn$ is the best solution for the rounded instance whose cost is $cost(\optn)$.  Let $\opt'$ be the best solution among the nice solutions which allocate no job to machines in 
$M(t_{min}) \setminus \{ h\}$.  Then, we next show that the cost $cost(\opt')$ of $\opt'$ is close to $cost(\optn)$.

\begin{lemma}
We have $cost(\optn) \leq cost(\opt') \leq (1+\eps)\cdot cost(\optn)$.
\end{lemma}
\begin{proof}
The first inequality holds by definition.  We prove the second inequality by establishing a nice solution $\sol$ which allocates no job to machines in 
$M(t_{min}) \setminus \{ h\}$ and we bound its cost by $(1+\eps) \cdot cost(\optn)$.
Consider \optn\ and let $J$ be the set of jobs which \optn\ assigns to machines in $M(t_{min}) \setminus \{ h\}$.
In order to construct \sol\ we modify $\optn$ by changing the allocation of $J$ and we assign these jobs to machine $h$ (and recall that $\sigma_h=1$).  Clearly by moving jobs from machines to machine $h$ the property of nice solutions cannot be hurt as the load in \sol\ of every machine  which is not $h$ is not larger than it was in \optn. Furthermore \sol\ does not allocate jobs to machines in $M(t_{min}) \setminus \{ h\}$.  Last, the cost of $\sol$ denoted as $cost(\sol)$ satisfies 
\begin{eqnarray}
cost(\sol) &\leq & cost(\optn) + \sum_{i\in M(t_{min}) \setminus \{ h\}} f_i \label{ineq1}\\
&\leq & cost(\optn) + \eps \cdot \sum_{i=1}^m f_i \label{ineq2}\\
&\leq & (1+\eps) \cdot cost(\optn) , \label{ineq3}
\end{eqnarray}
where (\ref{ineq1}) holds as for every machine $i\in M(t_{min}) \setminus \{ h\}$ that $\optn$ assigns a total load of $x$ (for $x\geq 0$) its cost in $\optn$ was at least $\sigma_i \cdot x \geq x$ and  this extra load on machine $h$ increases the cost of $h$ by at most $x$, and thus the increase of the cost of the solution due to the reallocation of jobs of size $x$ to machine $h$ is at most $f_i$ (the cost of $i$ in \sol); inequality (\ref{ineq2}) holds by definition of $t_{min}$; and (\ref{ineq3}) follows by the definition of the objective function of \pr.   
\qed \end{proof}

In what follows, we enforce the algorithm to allocate no job to machines in $M(t_{min}) \setminus \{ h\}$.  Since $\opt'$ is an optimal nice solution subject to this additional constraint, we conclude that it suffices to construct a feasible solution \sol\ whose cost is approximately $cost(\opt')$.  Next, we delete the set of machines $M(t_{min}) \setminus \{ h\}$ from the instance.  This deletion of machines does not hurt the feasibility of \sol\ and of $\opt'$ (as these solutions do not allocate jobs to the deleted machines), however, it decreases the cost of both solutions by a common non-negative constant that is the total fixed cost of the deleted machines.  Thus, it suffices to show that we can design an EPTAS for the instance resulted from $I$ by deleting the machines in $M(t_{min}) \setminus \{ h\}$.  Once again, with a slight abuse of notation we assume that the instance $I$ is the instance resulted from this deletion of machines and denote by $\{ 1,2,\ldots ,m\}$ the set of machines in $I$  such that $c_1\geq c_2 \cdots \geq c_m$.

We partition the machine set of the instance $I$ resulting from the deletion of $M(t_{min}) \setminus \{ h\}$.  This partition is obtained by letting each partition $S$ be a maximal (with respect to inclusion) set of consecutive indices of machines such that there are no two consecutive indices $i,i+1\in S$ satisfying $\frac{c_i}{c_{i+1}} \geq \left( \frac{1}{\eps} \right)^{2}$.  We let $\kappa$ be the number of partitions in this partition $(S_1,\ldots ,S_{\kappa})$ such that for every $q$, and for every $i(q) \in S_q$ and $i(q+1)\in S_{q+1}$ we have $c_{i(q)}>c_{i(q+1)}$ and in fact we have   $\eps^2 \cdot c_{i(q)}\geq c_{i(q+1)}$, by the sorting of the machines. For $q=1,2,\ldots ,\kappa$, let $\ell(q)=\min \{ i: i\in S_q\}$ and $r(q)=\max \{ i : i\in S_q\}$ so the indices in $S_q$ are those between $\ell(q)$ and $r(q)$.  A crucial property for our algorithm is the following one.

\begin{lemma}\label{bounded-ratio}
For every $q=1,2,\ldots ,\kappa$, and every pair of machines $i,i'\in S_{q}$, we have 
\begin{equation*}
\frac{c_i}{c_{i'}} \leq \left( \frac{1}{\eps} \right)^{4/\eps} .
\end{equation*}
\end{lemma}
\begin{proof}
Assume by contradiction that the claim does not hold for $i,i'\in S_q$.  Then, $\frac{c_i}{c_{i'}} > \left( \frac{1}{\eps^2} \right)^{2/\eps}$.  Then, $\log_{\frac{1}{\eps^2}} c_i - \log_{\frac{1}{\eps^2}} c_{i'}  > \frac{2}{\eps}.$  By the integrality of $\frac{2}{\eps}$, we conclude that the following holds.
\begin{equation*}
\left\lceil \log_{\frac{1}{\eps^2}} c_i \right\rceil - \left\lceil \log_{\frac{1}{\eps^2}} c_{i'} \right\rceil  \geq \frac{2}{\eps}.
\end{equation*}
Thus, by the pigeonhole principle, there are at least two integers $x<y$ which are equivalent to $t_{min}$ modulo $\frac{1}{\eps}$ such that $\left\lceil \log_{\frac{1}{\eps^2}} c_i \right\rceil \geq y >x > \left\lceil \log_{\frac{1}{\eps^2}} c_{i'} \right\rceil$.  Next, we define $z$ to be either $x$ or $y$ according to the following rule.  If $\left\lceil \log_{\frac{1}{\eps^2}} c_h \right\rceil \neq x$ then we let $z=x$, and otherwise we let $z=y$.  Then, observe that when we deleted the set of machines $M(t_{min}) \setminus \{ h\}$, we deleted all machines with standard capacities in the interval 
$ \left( \left( \frac{1}{\eps^2}\right)^{z-1}, \left( \frac{1}{\eps^2}\right)^{z} \right] $ and in particular $c_i$ and $c_{i'}$ do not belong to this interval.
Let $i''$ be the maximum index of a machine with $c_{i''} > \left( \frac{1}{\eps^2}\right)^{z}$.  Then, since $x\leq z\leq y$ by definition of $i''$ we have $c_i\geq  c_{i''} > c_{i'}$, however the ratio between $c_{i''}$ and $c_{i''+1}$ is strictly larger than $\frac{1}{\eps^2}$ contradicting the assumption that $i,i'\in S_q$ so $i,i',i'',i''+1\in S_q$, and thus the claim follows.
\qed\end{proof}

We next partition the job set $J=\{ 1,2,\ldots ,n\}$ as follows.  For $q=1,2,\ldots ,\kappa$ the job subset $J_q$ is defined as 
\begin{equation*} 
J_q = \{ j \in J: \frac{c_{\ell(q)}}{\eps} \geq p_j > \eps\cdot c_{r(q)} \}
\end{equation*}
The set $J_0$ is $J_0 =\{ j\in J: p_j > \frac{c_1}{\eps}\}$ and for every $q=1,2,\ldots ,\kappa$ the set 
$$J'_q=\{ j\in J: \frac{c_{\ell(q+1)}}{\eps}  < p_j \leq  \eps\cdot c_{r(q)}\}$$ where $c_{\ell(\kappa+1)}=0$.
Observe that every nice solution allocates all jobs of $J_0$ to machine $h$, and for every $q$ it allocates all jobs of $J_q\cup J'_q$ to machines in $\{ h\} \cup \bigcup_{i=1}^q S_i$.

For $q=1,2,\ldots ,\kappa$, we define a relaxation of the problem \pr\ where the set of machines is $S_q$, the set of jobs is $J_q$ and in addition we have sand consisting of jobs of total size $\phi_q$, and where we need to schedule all jobs and the sand on the machines $S_q$ but we are allowed to leave jobs and sand of total size at most $\psi_{q}$ unscheduled (these jobs are assigned to machines with indices smaller than $\ell(q)$ or to machine $h$).  The notion of sand means that the jobs that are part of the sand can be assigned fractionally to machines.   We denote by $\aux_q(\phi_q,\psi_q)$ the relaxation corresponding to the index $q$ together with the two numerical parameters $\phi_q,\psi_q$.

We will show that if $\phi_q$ is an integer multiply of $\eps \cdot c_{r(q)}$ while $\psi_q$ is an integer multiply of $\frac{c_{\ell(q)}}{\eps}$, then $\aux_q$ can be approximated within a multiplicative factor of $(1+\eps)$ with time complexity that fits the assumptions of an EPTAS.  That is, we will prove the following theorem in the next section.
\begin{theorem}\label{auxalg}
There exists an algorithm \alg\ that given an instance of $\aux_q$ defined by $q,\phi_q,\psi_q$ such that $\phi_q$ is an integer multiply of $\eps \cdot c_{r(q)}$ while $\psi_q$ is an integer multiply of $\frac{c_{\ell(q)}}{\eps}$, \alg\ returns a $(1+\eps)$-approximated solution to $\aux_q$ and the time complexity of $\alg$ is upper bounded by $T(m,n,\frac{1}{\eps})$ where $T(m,n,\frac{1}{\eps})=O(\left( (1/\eps)^{{O(1/\eps^{10})}} \right) \cdot m^2\log^3m)$.
\end{theorem}

Before presenting the proof of Theorem \ref{auxalg}, we show that the existence of the algorithm \alg\ is sufficient to guarantee the existence of an EPTAS for \pr.  
\begin{theorem}\label{mainthm}
There is an algorithm with time complexity $O(n^2\cdot m \cdot T(m,n,\frac{1}{\eps}))$ that given the rounded instance returns a solution whose cost is at most $(1+\eps)^4(1+3\eps)\cdot cost(\optn)$. 
\end{theorem}
\begin{proof}
It suffices to construct a $(1+\eps)^3(1+3\eps)$-approximation algorithm for the rounded instance after deleting the machines in $M(t_{min}) \setminus \{ h\}$.  

The first step of the algorithm is to apply \alg\ on a family $\cal F$ of inputs consisting of the following ones. For every $q=1,2,\ldots ,\kappa$, for every $\phi_q$ in the interval $[0,n \cdot \eps \cdot c_{r(q)}]$ that is an integer multiply of $\eps \cdot c_{r(q)}$, and for every $\psi_q$ in $[0,  n\cdot \frac{c_{\ell(q)}}{\eps}]$ that is an integer multiply of  $\frac{c_{\ell(q)}}{\eps}$, we apply \alg\ to solve approximately the instance $\aux_q(\phi_q,\psi_q)$.  We denote by $A(q,\phi_q,\psi_q)$ the cost of the solution $\sol(q,\phi_q,\psi_q)$ returned by \alg\ when applied on the instance $\aux_q(\phi_q,\psi_q)$.  The time complexity of the first step is 
$O(n^2m \cdot T(m,n,\frac{1}{\eps}))$ as the number of inputs solved by \alg\ is at most $O(n^2m)$ using  $\kappa \leq m$.

The second step is to use dynamic programming in order to concatenate a sequence of inputs in the family $\cal F$ consisting of one input for each value of $q$.  We define the dynamic programming formulation as a shortest path computation in a directed layered graph $G=(V,E)$.  The graph consists of $\kappa+1$ layers denoted as $L_0,L_1,L_2,\ldots ,L_{\kappa}$ and one additional node $t$.  The nodes of layer $L_q$ are associated with the possible value of $\phi_q$.  This defines the nodes of layers $L_1,L_2,\ldots ,L_{\kappa}$, however to use this definition for layer $L_0$ we define a value $\phi_0$ as the total size of jobs in $J_0$ plus the value of $\psi_1$, i.e., $\phi_0=\sum_{j\in J_0} p_j +\psi_1$. Thus, in every layer  there are $n+1$ nodes, and in total there are $O(nm)$ nodes in $G$.  We next describe the edge set of $G$ together with the length associated with each edge.  For $q=\kappa,\kappa-1,\ldots ,1$ and a node $\phi_q$ in layer $L_q$ and node $\phi_{q-1}$ in layer $L_{q-1}$ we have an edge from the node $\phi_q$ in $L_q$ towards node $\phi_{q-1}$ in $L_{q-1}$ whose length is defined as follows.  We compute a value $\psi_q$ that is the maximum integer multiply of $\frac{c_{\ell(q)}}{\eps}$ such that together with the total size of jobs in $J'_{q-1}$ the resulting size is at most $\phi_{q-1}$.  That is, for $q\geq 2$, we compute \begin{equation}\label{psiq-def}
\psi_q= \frac{c_{\ell(q)}}{\eps} \cdot \min\left\{n, \left\lfloor \frac{\phi_{q-1} - \sum_{j\in J'_{q-1}} p_j}{\frac{c_{\ell(q)}}{\eps}} \right\rfloor \right\} .
\end{equation} 
The value of $\psi_1$ is computed slightly different.  We subtract from $\phi_0$ the total size of jobs in $J_0$ and the resulting value is $\psi_1$ (note that this is already a rounded value).  That is, $\psi_1=\phi_0-\sum_{j\in J_0} p_j$.
 The length of the edge in the graph between these two nodes is defined as $A(q,\phi_q,\psi_q)$.  For every node $\phi_0$ in layer $L_0$ we have an edge from this node directed to $t$ whose length is $\phi_0$.  This length of the edges directed to $t$ is motivated by the fact that assigning jobs of total size $\phi_0$ to machine $h$ costs at most $\phi_0$.  In the resulting directed graph we find a shortest path $\cal P$ from the node $\phi_{\kappa}$ in layer $L_{\kappa}$ to node $t$, where $\phi_{\kappa}$ is defined as follows.  
\begin{equation*}
\phi_{\kappa} = \eps\cdot c_{r(\kappa)} \cdot \left\lceil  \frac{\sum_{j \in J'_{\kappa}}p_j}{\eps\cdot c_{r(\kappa)}}  \right\rceil .
\end{equation*}
The time complexity of the second step is determined by the number of edges in the graph that is at most $O(n^2m)$.  We denote by $\phi^*_q$ the node in layer $L_q$ that belongs to the shortest path computed by the algorithm, and we let $\psi^*_q$ be the corresponding value of $\psi_q$ that the algorithm computed using \eqref{psiq-def} for the sequence of $\phi^*$.

The third (and last) step of the algorithm is to compute a feasible solution for \pr\ whose cost is at most $(1+\eps)$ times the total length of $\cal P$.  For $q=\kappa,\kappa-1,\ldots ,1$, we show that we can assign (integrally) the jobs $J_q$ and small jobs of total size $\phi^*_q$ each of which of size at most $\eps\cdot c_{r(q)}$ such that a total size of at most $\psi^*_q$ is not assigned (such a solution is called feasible), and the cost of this feasible solution is at most $(1+\eps) \cdot A(q,\phi^*_q,\psi^*_q)$.  

Consider one specific value of $q$.  We say that the jobs in $J_q$ are {\it large} and the other jobs are {\it small}.  The solution $\sol(q,\phi^*_q,\psi^*_q)$ returned by \alg\ specifies the assignment of large jobs to machines in $S_q$ (some of these jobs might be unassigned) and for each machine $i\in S_q$ it defines a volume $vol(i)$ of sand that is assigned to $i$.  We denote by $J(q)\subseteq J_q$ the set of large jobs that the solution $\sol(q,\phi^*_q,\psi^*_q)$  does not assign to machines in $S_q$.  The feasibility of the solution $\sol(q,\phi^*_q,\psi^*_q)$  (for $\aux_q$) ensures the following inequalities. 
\begin{equation}\label{vol-guarantee}
\phi^*_q - \sum_{i\in S_q} vol(i) + \sum_{j\in J(q)} p_j \leq \psi^*_q, \ \ \ \ \ \ \mbox{and} \ \ \ \ \ \phi^*_q \geq \sum_{i\in S_q} vol(i) ,  
\end{equation} 
where the first inequality holds by the guarantee on the total size of jobs and sand that the solution does not assign, and the second inequality follows by the fact that the total size of sand in the instance is at most $\phi^*_q$.  In the solution that we create we assign the jobs in $J_q\setminus J(q)$ exactly as in $\sol(q,\phi^*_q,\psi^*_q)$, while for the assignment of small jobs we consider the list of small jobs $Small$ and we process the machines in $S_q$ one by one in an arbitrary order as long as $Small$ is not empty.  When considering the current machine $i$, we find a minimum prefix of $Small$ whose total size is at least $vol(i)$, this prefix of jobs is assigned to $i$, we delete it from $Small$ and move to the next machine in $S_q$.  If this prefix is undefined, it means that the total size of jobs in $Small$ is smaller than $vol(i)$ and we assign all jobs in $Small$ to $i$ (and stop the assignment process of small jobs to machines in $S_q$).  The time complexity of this step is $O(n+m)$.  Furthermore, if there is a machine $i$ such that when processing $i$ all jobs in $Small$ are assigned to $i$, then all small jobs are assigned and the feasibility of the solution we create to machines in $S_q$ follows by \eqref{vol-guarantee}.  Otherwise, every machine $i\in S_q$ receives a total size of small jobs of at least $vol(i)$, and once again by \eqref{vol-guarantee} the resulting solution we create is a feasible solution.  We observe that for every machine $i\in S_q$, the total size of jobs assigned to $i$ is at most $\eps\cdot c_{r(q)} \leq \eps c_i$ larger than the total size of jobs (and sand) assigned to $i$ in  $\sol(q,\phi^*_q,\psi^*_q)$.  By observation \ref{obs1}, this increase of the total size of jobs assigned to $i$ may increase the cost of $i$ by a multiplicative factor of at most $(1+\eps)$ as we show next.  If $x$ denotes the total size of jobs and sand assigned to $i$ in $\sol(q,\phi^*_q,\psi^*_q)$, then the cost of $i$ in that solution is $cost_i(\max\{ x,c_i\})$ and in our solution it is at most $cost_i(x+\eps \cdot c_i) \leq cost_i((1+\eps) \cdot \max\{ x,c_i\}) \leq (1+\eps) \cdot cost_i(\max\{ x,c_i\})$ where the first inequality follows by the monotonicity of the cost function and the second inequality by Observation \ref{obs1}.  In order to use the induction (and decrease the value of $q$ by $1$), note that the total size of jobs not assigned to machines that are not $h$ and with indices at least $\ell(q)$ which are of size at most $\frac{c_{\ell(q)}}{\eps}$ is at most $\psi^*_q$, and the total size of jobs with sizes in the interval $( \frac{c_{\ell(q)}}{\eps} ,  \eps\cdot c_{r(q-1)} ] $ is the total size of jobs in $J'_{q-1}$.  Thus, by the definition of $\psi^*_q$ in terms of $\phi^*_{q-1}$, we conclude that the total size of jobs of size at most $\eps\cdot c_{r(q-1)}$ that are still unscheduled is at most $\phi^*_{q-1}$ and indeed we guarantee the assumption on the recursive algorithm for $q-1$.  The claim follows as any set of jobs of total size at most $\phi^*_0$ can be assigned to machine $h$ increasing the cost of that machine by at most $\phi^*_0$ that is the length of the edge of $\cal P$ adjacent to $t$.

The theorem follow by showing that the graph $G$ has a path $P_{opt}$ whose total length is at most $(1+\eps)\cdot(1+3\eps)\cdot cost(\opt')$ where $\opt'$ is a cheapest solution among all nice solutions which do not allocate jobs to machines in $M(t_{min}) \setminus \{ h\}$.  
Based on \opt' we define a fractional value of $\phi_q$ for all $q=0,1,2,\ldots ,\kappa$ as follows where we let $J'_0=J_0$.  For a given value of $q$, the fractional value $\hat{\phi}_q$ of $\phi_q$ is the total size of jobs in $J'_q\cup \bigcup_{q'=q+1}^{\kappa} (J_{q'}\cup J'_{q'})$ that \opt' assigns to machines in $\{ h \} \cup \bigcup_{q'=1}^{q} S_{q'}$.   Similarly, we define $\hat{\psi}_q= \hat{\phi}_{q-1} - \sum_{j\in J'_{q-1}} p_j$.
By the definition of $\aux_q$ we conclude that the cost of an optimal solution to $\aux_q(\hat{\phi}_q,\hat{\psi}_q)$ is at most the cost \opt' pays for machines in $S_q$.

Next, for every $q=1,2,...,\kappa$, we round up $\hat{\psi}_q$ to the next integer multiply of $\eps\cdot c_{r(q-1)}$ and we denote by $\psi'_q$ this rounded up value. This may force us to increase $\phi_{q-1}$ and thus our next step is to round up $\hat{\phi}_q$ (for all $q=1,2,\ldots ,\kappa$) to the next integer multiply of $\eps\cdot c_{\ell(q)}$ and to add another $\eps\cdot c_{\ell(q)}$ to the rounded up value to get the value $\phi'_q$.  The rounding of $\phi'_0$ is different and we round down $\hat{\phi}_0$ to the next value of the form of the total size of jobs in $J_0$ plus an integer multiply of $\eps \cdot c_{\ell(1)}$.

When we compare the two instances (for $q\geq 1$) of the auxiliary problem $\aux_q(\hat{\phi}_q,\hat{\psi}_q)$ with $\aux_q(\phi'_q,\psi'_q)$, we can take a solution of the first one and add $3\eps\cdot c_{\ell(q)}$ size of sand to machine $c_{\ell(q)}$ to get a feasible solution of the second problem. This is sufficient even for $q=1$ to get a feasible solution for the instance we solved for the edge between $\phi'_1$ in layer $L_1$ to node $\phi'_0$ in layer $L_0$.  This additional sand increases the cost of machine $\ell(q)$ by a multiplicative factor of at most $(1+3\eps)$ but this input satisfies the assumptions for which \alg\ is a $(1+\eps)$-approximation for \aux.  Thus, the length of the edge between $\phi'_q$ in layer $L_q$ to $\phi'_{q-1}$ in layer $L_{q-1}$ is at most $(1+\eps)\cdot(1+3\eps)$ times the total cost of the machines in $S_q$ that \opt' pays.   Since $\phi'_0$ is smaller than $\hat{\phi}_0$, we conclude that the total size of jobs which \opt' assigns to $h$ is larger than the one in our solution due to this edge directed to $t$.
\qed \end{proof}

\section{Approximating  $\aux_q(\phi_q,\psi_q)$ via the use of $n$-fold programming\label{nfoldsec}}
We assume that $\phi_q$ is an integer multiply of $\eps \cdot c_{r(q)}$ while $\psi_q$ is an integer multiply of $\frac{c_{\ell(q)}}{\eps}$ (and hence also an integer multiply of $\eps \cdot c_{r(q)}$).
We first show that by restricting ourselves to solutions of $\aux_q$ for which the total size of sand assigned to each machine is an integer multiply of $\eps c_{r(q)}$ the approximation ratio is multiplied by at most $1+\eps$.  We denote by $\aux'_q(\phi_q,\psi_q)$ the resulting auxiliary problem with this additional constraint.
\begin{lemma}\label{apx-n-fold-lem}
Let $\sol'$ be an optimal solution for $\aux'_q(\phi_q,\psi_q)$, then $\sol'$ is a $(1+\eps)$-approximation for  $\aux_q(\phi_q,\psi_q)$.
\end{lemma}
\begin{proof}
$\sol'$ is clearly a feasible solution to $\aux_q(\phi_q,\psi_q)$.  It thus suffices to upper bound its cost.  Let $\sol$ be an optimal solution for $\aux_q(\phi_q,\psi_q)$.  We modify the (total) size of sand assigned to each machine $i\in S_q$ by rounding it up to the next integer multiply of $\eps \cdot c_{r(q)}$.  If the total size of sand which we allocate is larger than $\phi_q$, then we decrease integer multiplies of  $\eps \cdot c_{r(q)}$ from the size of sand assigned to some machines so that the total size of sand which we assigned is exactly $\phi_q$.  Observe that by rounding up the size of sand assigned to each machine we increase its load by at most $\eps \cdot c_{r(q)}$.  Let $x$ be the original load of $i$ (in $\sol$) and let $y$ be its load in the new created solution, then we have $x\leq y \leq x+\eps c_{r(q)}$.  If $y\leq c_i$, then the cost of machine $i$ is $f_i$ in both solutions.  Otherwise, $cost_i(y) \leq cost_i(x)+ \eps f_i\leq (1+\eps)cost_i(x)$.  Thus, the cost of $\sol'$ is at most $1+\eps$ times the cost of $\sol$.
\qed\end{proof}

Based on Lemma \ref{apx-n-fold-lem}, the proof of Theorem \ref{auxalg} and thus also the proof of Theorem \ref{mainthm} follow by establishing an exact algorithm for solving  $\aux'_q(\phi_q,\psi_q)$ (i.e., an algorithm for finding an optimal solution of $\aux'$) whose time complexity is upper bounded by $O\left(\left( (1/\eps)^{{O(1/\eps^{10})}} \right) \cdot m^2\log^3 m \right)$ like the algorithm we present next.

The first step of the algorithm is to partition the sand of size $\phi_q$ into a set of $\frac{\phi_q}{\eps\cdot c_{r(q)}}$ {\em dummy jobs} each of which of size $\eps c_{r(q)}$.  Observe that the total size of these dummy jobs is $\phi_q$ and an assignment of the jobs in $J_q$ and the dummy jobs to machines in $S_q$ such that the total size of unassigned jobs and dummy jobs is at most $\psi_q$ is a feasible solution to $\aux'$ and this is a characterization of the feasible solutions of $\aux'$.  Let $J^q$ be the set of jobs and dummy jobs of this instance.  In what follows we say a job $j$ and mean that $j$ is either a job or a dummy job, that is, we do not distinguish between jobs and dummy jobs of the same size.
For every $p$ that is a size of a job in $J^q$, we denote by $n_p$ the number of jobs of $J^q$ of size $p$.

Note that all jobs in $J^q$ have sizes that are integer multiply of $\eps c_{r(q)}$ and have sizes of at most $\frac{c_{\ell(q)}}{\eps}$.  Due to our rounding of the job sizes there are $O(1/\eps^2)$ distinct sizes in every interval of sizes where the upper bound is at most $1/\eps$ times the lower bound of the interval.  Thus, the number of distinct sizes of jobs in $J^q$ is at most $\frac{1}{\eps^2} \cdot \left( \log_{1/\eps} \frac{c_{\ell(q)}}{c_{r(q)}} +1 \right) < \left( \frac{7}{\eps^4} \right)$ where the inequality follows by lemma \ref{bounded-ratio}.  We let $B_q$ be the set of distinct sizes of jobs in $J^q$.

For machine $i\in S_q$ we define a configuration of machine $i$ as a vector consisting of $|B_q|$ components where the components are associated with the elements in $B_q$ in increasing order (of the sizes in this set).  Each component corresponding to $p\in B_q$ represents the number of jobs of size $p$ which are assigned to $i$.  This number is a non-negative integer that is at most $\frac{1}{\eps^2} \cdot \frac{c_{\ell(q)}}{c_{r(q)}} < \left( \frac{1}{\eps} \right)^{7/\eps}$ (as the load of $i$ is at most $\frac{c_{\ell(q)}}{\eps}$).  Thus, the number of distinct configurations of machine $i$ is at most 
$\left( \frac{1}{\eps} \right)^{(7/\eps)\cdot \left( \frac{7}{\eps^4} \right)} = \left( \frac{1}{\eps} \right)^{(49/\eps^5)} $.   Each such configuration of machine $i$ has a cost that is the value of $cost_i$ when assigned the set of jobs of this configuration.  We denote by ${\cal C}_i$ the set of configurations of machine $i$, and for each $c\in {\cal C}_i$ we let $cost^i(c)$ be the cost of this configuration.

We next formulate $\aux'$ as an integer linear program.  The decision variables are $x_{i,c}$ for every machine $i\in S_q$ and $c\in {\cal C}_i$ that is an indicator variable that equals $1$ when machine $i$ is assigned a configuration $c$ and $0$ otherwise, and the set of additional variables $y_p$ for every $p\in B_q$ encoding the number of jobs in $J^q$ of size $p$ which are not assigned to machines in $S_q$.

The objective function is clearly to minimize the total cost of the used configurations.  That is, 
\begin{equation}\label{objective}
\min \ \sum_{i\in S_q} \sum_{c\in {\cal C}_i} cost^i(c) \cdot x_{i,c} .
\end{equation}

We have the following families of constraints:

\paragraph{The global constraints.}  We have a constraint for each $p\in B_q$ saying that every job of size $p$ is either assigned to one of the machines in $S_q$ or not assigned to any machine in $S_q$.  That is, for every $p\in B_q$ we introduce the constraint:
\begin{equation}\label{global-cons}
\sum_{i\in S_q}\sum_{c\in {\cal C}_i}  c_p \cdot x_{i,c} + y_p = n_p .
\end{equation}   
In addition we have a bound of $\psi_q$ on the total size of unassigned jobs.  We divide this inequality by $\eps c_{r(q)}$ and obtain the following inequality as an additional global constraint.
\begin{equation}\label{last-global-cons}
\sum_{p\in B_q} \frac{p}{\eps c_{r(q)}} \cdot y_p \leq \frac{\psi_q}{\eps c_{r(q)}} .
\end{equation}
We use the division by this common factor to conclude that all coefficients of the global constraints are non-negative integers which are at most  $ \left( \frac{1}{\eps} \right)^{7/\eps}$.
Furthermore, observe that the number of global constraints which we denote by $r$ is a small constant $r=|B_q|+1 \leq  \left( \frac{7}{\eps^4} \right)$.

Next, we group the decision variables $x_{i,c}$ in {\em bricks} where a brick is the collection of variables corresponding to one specific machine $i$.  The columns corresponding to variables of each brick are consecutive columns of the resulting constraint matrix.

\paragraph{The local constraints.} For every brick, namely for every machine $i$, we have one local constraint involving (only) variables of that brick, namely the constraint that each machine is assigned exactly one configuration.  That is, for every $i$, the local constraint of brick $i$ is 
\begin{equation}\label{local-cons}
\sum_{c\in {\cal C}_i} x_{i,c} =1 .
\end{equation}

In addition, we have lower and upper bounds on the variables.  In our settings the $x_{i,c}$ is an indicator variable (so it should be between $0$ and $1$) while the $y_p$ are non-negative and we can add the additional (meaningless) upper bound of $n$.  Thus we introduce the following bounds.

\begin{equation}\label{variables-bounds}
0 \leq x_{i,c} \leq 1 \ \ \ \forall i\in S_q, \ \forall c\in {\cal C}_i \ \ \ \ \ \mbox{ and } \ \ \ \ \ 0 \leq y_p \leq n \ \ \ \forall p \in B_q .
\end{equation}

Using these constraints and variables, the integer linear program formulating $\aux'$ is to minimize the objective \eqref{objective} subject to the constraints \eqref{global-cons} for every $p\in B_q$, the constraint \eqref{last-global-cons}, the constraints \eqref{local-cons} for every $i\in S_q$, and the constraints \eqref{variables-bounds} (in addition to the requirement that all variables are integers).

For using the results for $n$-fold programming we note the following bounds. 
\begin{itemize}
\item The number of global constraints is $r \leq  \frac{7}{\eps^4}$.
\item The number of local constraints of a brick is $s=1$.
\item The maximum absolute value of a component of the constraint matrix (i.e., the infinity-norm of the matrix) is $a\leq  \left( \frac{1}{\eps} \right)^{7/\eps}$. 
\item The number of variables in every brick is $t \leq \left( \frac{1}{\eps} \right)^{(49/\eps^5)} $. 
\item The number of bricks is $|S_q|\leq m$ while the right hand side is bounded by $n$.
\end{itemize}

The problem we formulated is a special case of {\em generalized $n$-fold programming} where the number of variables $N = |S_q| \cdot t+|B_q| \leq mt+\frac{7}{\eps^4} \leq mt+t$.  The running time of the algorithm of Eisenbrand et al. \cite{E+19} for solving such problem (see the scaling,$\rho$ column of the linear objective case of Corollary 91 in \cite{E+19}) is upper bounded by $O((ars)^{O(r^2s+s^2r)} \cdot N^2 \log^3 (nN))$.  Using our bounds and $s=1$, this is upper bounded by $O((ar)^{O(r^2)} \cdot (t^2\log^3 t) \cdot (m^2\log^3m))$.  The coefficient $(ar)^{O(r^2)} \cdot t^2\log^3 t$ is a function of $\eps$ that is upper bounded by $O\left( (1/\eps)^{{O(1/\eps^{10})}} \right)$ and $m^2\log^3m$ is a strongly polynomial bound independent of $\eps$.

\section{An asymptotic fully polynomial time approximation scheme for \prc \label{prcsec}}
When considering asymptotic schemes that return a solution of cost at most $(1+\eps) cost(\opt(I)) + g(1/\eps)$ for some function $g$ where $cost(\opt(I))$ is the optimal cost of the same instance, the additive term $g(1/\eps)$ is not scalable. In order to use this definition of asymptotic approximation scheme we  assume that $\max_{i=1,2,\ldots, m} f_i =1$.

\begin{lemma}\label{fmaxlem}
Without loss of generality $\max_{i} f_i =1$, and $\min_i \sigma_i =1$.
\end{lemma}
\begin{proof}
Assume that the claim does not hold.  We first scale all the $f_i,\sigma_i$ by dividing these numbers by a common factor of $\max_i f_i$.  Observe that the cost of every feasible solution is scaled by this factor, and the first part of the claim holds.  Furthermore, assumption  \eqref{assump1} still holds.  We apply Lemma \ref{sigma-min} and observe that the transformation in the proof of that lemma does not change the fixed costs of machines.
\qed
\end{proof}

Let $I$ be the resulting instance after (perhaps) changing the input according to the proof of the last lemma.  Let $\opt(I)$ denotes an optimal solution for this instance, and $cost(\opt(I))$ denotes its cost.  We  show the existence of an algorithm that returns a feasible solution for $I$ with cost at most  $(1+\eps) \cdot (1+2\eps) \cdot cost(\opt(I)) + g(1/\eps)$ and time complexity that is upper bounded by a polynomial in $n,m,\frac{1}{\eps}$.

\begin{lemma}\label{sollem1}
Consider an optimal solution $\sol$ for $I$ under the additional constraint that for every $i$ and  every machine $\mu$ of type $i$, either $\sol$ assigns a unique job to $\mu$, or the load of $\mu$ is at most $\frac{c_i}{\eps}$.  Then, the cost of $\sol$ is at most $(1+2\eps)\cdot cost(\opt(I))$.
\end{lemma}
\begin{proof}
Assume that the assumption is not satisfied with respect to the set of jobs $J'$ that $\opt(I)$ assigns to $\mu$.  We  replace $\mu$ by a collection of machines of type $i$ and assign $J'$ to these machines such that the total cost of these machines is at most $(1+2\eps)$ times the cost of $\mu$ in $\opt(I)$.  Applying this transformation for each machine establish the claim of the lemma.

Consider $J'$, for every job of size at least $\frac{c_i}{\eps}$ we add a dedicated machine of type $i$ for this job and assign it to its dedicated machine.  Performing this step on all such jobs may increase the cost of the solution by at most an additive term of $f_i$ but this additive term is at most an $\eps$ fraction of the cost of $\mu$ in $\opt(I)$. Moreover, if an increase of the cost occurred it means that the resulting set of machines (the new dedicated machines as well as machine $\mu$) satisfy the conditions of the lemma.

If the conditions of the claim do not hold yet, then in particular it means that the remaining set of  jobs $J''\subseteq J'$ that were not assigned to dedicated machines (by the previous modification) have total size larger than  $\frac{c_i}{\eps}$.  Then, we pack these jobs into bins, each of which of capacity $\frac{c_i}{\eps}$ using the next-fit heuristic.  That is, we have an {\em open machine} of type $i$, and we process the jobs one by one. When we process job $j$, we try to assign it to the current open machine.  If the resulting set of jobs assigned to the open machine has total size at most $\frac{c_i}{\eps}$ we do so, and continue to the next job,  otherwise we close the open machine and open a new open machine of type $i$ and assign $j$ there.  If $L$ was the total size of jobs in $J''$ we use at most $\frac{2\eps L}{c_i}$ machines to pack all the jobs in $J''$, and this may increase the cost of the resulting solution by the total fixed cost of these machines, that is, by at most $\frac{2\eps L}{c_i}\cdot f_i = 2\eps L \sigma_i$ using \eqref{assump1}. This is at most $2\eps$ times the cost of assigning $J''$ to $\mu$ and the claim follows.
\qed
\end{proof}

Let $\zeta$ be such that $c_{\zeta}=\max_i c_i$.  Then, by Lemma \ref{sollem1}, we conclude that every job of size larger than $\frac{c_{\zeta}}{\eps}$ is allocated a dedicated machine.  For each such job, we find the type $i$ for which the resulting cost of the dedicated machine is minimized and we use such machine to process the job.  In this way, we  eliminate all jobs of size larger than $\frac{c_{\zeta}}{\eps}$.  In the remaining instance that we denote as $I'$, the load of every machine (in \sol) is at most $\frac{c_{\zeta}}{\eps}$ and for every collection of such jobs there is a type of machines such that if we assign it to such machine, then the resulting cost would be at most $\frac{c_{\zeta}}{\eps}\cdot \sigma_{\zeta} = \frac{f_{\zeta}}{\eps} \leq \frac{1}{\eps}$ where the inequality follows by Lemma \ref{fmaxlem}.

It suffices to construct an asymptotic approximation scheme for $I'$ where we modify the definition of the problem so that the load of every machine is at most  $\frac{c_{\zeta}}{\eps}$.  We next show that such a scheme was established by Epstein and Levin in \cite{EL17}.   To use the results of \cite{EL17}, we transform the instance $I'$ into an instance of {\em bin packing with bin utilization cost (BPUC)} for which \cite{EL17} designed an AFPTAS.  In BPUC we are given a monotonically
non-decreasing non-negative cost function $\pi$, where its domain contains the interval $[0,1]$, and a set of items $\{1,2,\ldots ,n\}$, where item $j$
has a non-negative size $s_j$ (such that $s_j \in [0,1]$ for all $j$). The goal is to partition the items
into subsets $S_1,\ldots ,S_m$ such that the total size of items
in each subset is at most $1$ and the cost, which is defined as
$\sum_{i=1}^m \pi (\sum_{j\in S_i} s_j)$, is minimized. 

In order to transform $I'$ into an instance of BPUC, we do the following.  The set of items is the set of jobs.  The size of item $j$ in the BPUC instance is $\frac{p_j}{c_{\zeta}/\eps}$ (that is, the fraction of a largest load of a machine in a solution that satisfies the additional constraint), and to define the bin utilization cost $\pi$ we do the following.  For $x\in [0,1]$, we set $\pi(x)$ to be
\begin{equation}\label{pieq}
\pi(x)= \min_{i=1,2,\ldots ,m} cost_i(x\cdot \frac{c_{\zeta}}{\eps}).
\end{equation}
Observe the following simple properties. First, $\pi$ is a monotone non-decreasing function as for every $i$  the cost function $cost_i$ is monotone non-decreasing. Second for every $x\in [0,1]$ we can evaluate $\pi(x)$ in polynomial time as we can evaluate $cost_i$ in constant time for every $i$. Last for every $y>0$ we can find a maximum value $x$ such that $\pi(x)\leq y$, since for every $i$ in constant time we can compute a maximum value of $x$ such that $cost_i(x) =y$ (if $y\geq f_i$ while if $y<f_i$ then there is no $x$ for which $cost_i(x) \leq y$).  These properties guarantee the assumptions used by \cite{EL17} to design their AFPTAS for BPUC.

The following theorem follows by the observation that partitioning the jobs of $I'$ to subsets according to the solution obtained for BPUC and then choosing for each subset the type of machine that minimizes the cost of assigning the jobs to that machine type, results in a solution for $I'$ of the same cost as the cost of the solution for BPUC.  This holds also in the other direction if we are given an optimal solution for $I'$ we obtain an optimal solution for the input of BPUC.  Thus, $I'$ is equivalent to the instance we created for BPUC which proves the following theorem.
  
\begin{theorem}
There is an AFPTAS for \prc.
\end{theorem}

\bibliographystyle{abbrv}

\end{document}